\documentclass[showpacs,amsmath,amssymb,twocolumn,superscriptaddress,notitlepage,preprintnumbers,pra]{revtex4-1}

\usepackage[dvips]{graphicx}
\usepackage{amsmath,amssymb,amsthm,mathrsfs,amsfonts,dsfont}
 
\usepackage{dcolumn}
\usepackage{bm}
\usepackage{amsmath}
\usepackage{amssymb}
\usepackage{braket}
\usepackage{physics}
\usepackage{amsthm}
\usepackage{natbib}
\usepackage{mathtools}
\usepackage{diagbox}
\usepackage[utf8]{inputenc}
\usepackage[english]{babel}
\usepackage{scalerel}[2014/03/10]
\usepackage{stackengine}
\usepackage{algorithm}
\usepackage[noend]{algpseudocode} 
\usepackage{makecell}
\usepackage{footnotebackref}
\usepackage{lipsum}

\usepackage{hyperref}
\hypersetup{colorlinks=true, linkcolor=blue, citecolor=blue, urlcolor=black }

\algrenewcommand\algorithmicrequire{\textbf{Input:}}
\algrenewcommand\algorithmicensure{\textbf{Output:}}

\makeatletter
\newcommand{\algmargin}{\the\ALG@thistlm}
\makeatother
\newlength{\whilewidth}
\algdef{SE}[parWHILE]{parWhile}{EndparWhile}[1]
  {\parbox[t]{\dimexpr\linewidth-\algmargin}{%
     \hangindent\whilewidth\strut\algorithmicwhile\ #1\ \algorithmicdo\strut}}{\algorithmicend\ \algorithmicwhile}%
\algnewcommand{\parState}[1]{\State%
  \parbox[t]{\dimexpr\linewidth-\algmargin}{\strut #1\strut}}

\newtheorem{theorem}{Theorem}[]
\newtheorem{corollary}{Corollary}[theorem]
\newtheorem{lemma}[theorem]{Lemma}

\theoremstyle{definition}

\newcommand{\bytedance}{ByteDance Ltd., Zhonghang Plaza, No. 43, North 3rd Ring West Road, Haidian District, Beijing, China}
\newcommand{\uwa}{Department of Physics, The University of Western Australia, Perth, WA 6009, Australia}
\newcommand{\oxford}{Clarendon Laboratory, University of Oxford, Parks Road, Oxford OX1 3PU, United Kingdom}

\begin{document}


\title{Orbital Expansion Variational Quantum Eigensolver: \\ Enabling Efficient Simulation of Molecules with Shallow Quantum Circuit}
\author{Yusen Wu}
  \thanks{The first two authors contributed equally.}
  \affiliation{\bytedance}
  \affiliation{\uwa}

\author{Zigeng Huang}
\email{huangzigeng@bytedance.com}
\affiliation{\bytedance}

\author{Jinzhao Sun}
\affiliation{\oxford}

\author{Xiao Yuan}
\affiliation{Center on Frontiers of Computing Studies, Peking University, Beijing 100871, China}
\affiliation{School of Computer Science, Peking University, Beijing 100871, China}

\author{Jingbo B. Wang}
\affiliation{\uwa}

\author{Dingshun Lv}
\thanks{lvdingshun@bytedance.com}
\affiliation{\bytedance}

\begin{abstract}
In the noisy-intermediate-scale-quantum era, Variational Quantum Eigensolver (VQE) is a promising method to study ground state properties in quantum chemistry, materials science, and condensed physics.  However, general quantum eigensolvers are lack of systematical improvability, and achieve rigorous convergence is generally hard in practice, especially in solving strong-correlated systems. Here, we propose an Orbital Expansion VQE~(OE-VQE) framework to construct an efficient convergence path. The path starts from a highly correlated compact active space and rapidly expands and converges to the ground state, enabling simulating ground states with much shallower quantum circuits. We benchmark the OE-VQE on a series of typical molecules including H$_{6}$-chain, H$_{10}$-ring and N$_2$, and the simulation results show that proposed convergence paths dramatically enhance the performance of general quantum eigensolvers.

\end{abstract}

\maketitle

\section{Introduction}

Solving ground states of quantum systems is a natural application for quantum computers, which could potentially innovate the study of quantum chemistry, materials science, and many-body physics~\cite{chan2011density,motta2020determining, mcardle2020quantum,cao2018quantum, wu2021provable}. However, quantum techniques like quantum phase estimation which promises accurate chemical simulations requiring fault-tolerant quantum computers~\cite{bravyi2018quantum}, is beyond current quantum computers. In order to reduce the significant hardware demands required by universal quantum algorithms, the variational quantum eigensolver~(VQE) was proposed~\cite{o2016scalable, whitfield2011simulation, mcclean2016theory, barkoutsos2018quantum, romero2018strategies,lee2018generalized} and demonstrated on noisy-intermediate scale quantum~(NISQ) devices~\cite{preskill2018quantum, peruzzo2014variational, colless2018computation, hempel2018quantum, kandala2017hardware, google2020hartree}. 
While quantum hardware continue to steadily advance, limitations on deep high-fidelity circuit still exist~\cite{stilck2021limitations, wu2022estimating}. Therefore, the demand of quantum algorithms that try to make full use of the limited quantum resource is growing rapidly. Particularly, one may think: how to efficiently design a powerful quantum ansatz allowing accurate computation with much shallower circuit depth. In other words, find an efficient convergence path from the fixed initial state (such as Hartree-Fock state) to the exact ground state. Actually, there are exponential large number of paths between two states in the Bloch sphere, and the length of the optimal convergence path defines the circuit complexity of the ground state~\cite{nielsen2006quantum}. 

From different perspectives, previous works consider various strategies to reduce the quantum circuit depth in solving quantum chemical problems, including qubit-reduction methods~\cite{cade2020strategies, mineh2022solving, fujii2022deep, li2022toward, zhang2022variational, cao2022ab, tilly2021reduced, eddins2022doubling}, circuit depth optimization methods~\cite{grimsley2019adaptive, zhang2020mutual, tang2021qubit2, huang2022efficient, burton2022exact} and unitary couple cluster with its variant methods~\cite{o2016scalable, mcclean2016theory, barkoutsos2018quantum, romero2018strategies, cao2022progress, ryabinkin2020iterative}. 
Although such methods can be utilized to reduce quantum circuit depth, they suffer from a lack of systematic improvability when constructing a convergence path, that is the energy function may not decline steadily with the increase of quantum (or classical) computational resources. For example, the famous heuristic method ADAPT-VQE which proposed in~\cite{grimsley2019adaptive} may encounter the scenario where the single gradient-based criteria cannot guarantee to find the global minimum. This because general unitary operators lie in an extremely complex manifold, and fixed optimization strategy may fail. 

Here, we propose an Orbital Expansion VQE~(OE-VQE) framework to provide an efficient convergence path in simulating molecular systems, which dramatically improves the performance of a general quantum eigensolver and reduces quantum measurement complexity. The proposed convergence path is composed of two fundamental elements: (i) a good initial state inspired by chemical insights and (ii) a systematic improvable convergence direction to the ground state. In the proposed OE-VQE framework, a specially initial state, which is constructed from a very small active space composed by orbitals with significant correlations, is selected as the starting point. The rest orbitals are reformulated and then ranked by the low-scaling post-Hartree-Fock methods, such as the second-order (Moller-Plesset) perturbation theory~(MP2)~\cite{nusspickel2022systematic}. The initial active space expands steadily by iteratively appending ranked orbitals, and a reasonable convergence direction is thus constructed. Following this path, the accuracy of quantum eigensolver will be systematically improved step by step.

\begin{figure*}
\includegraphics[width=\textwidth]{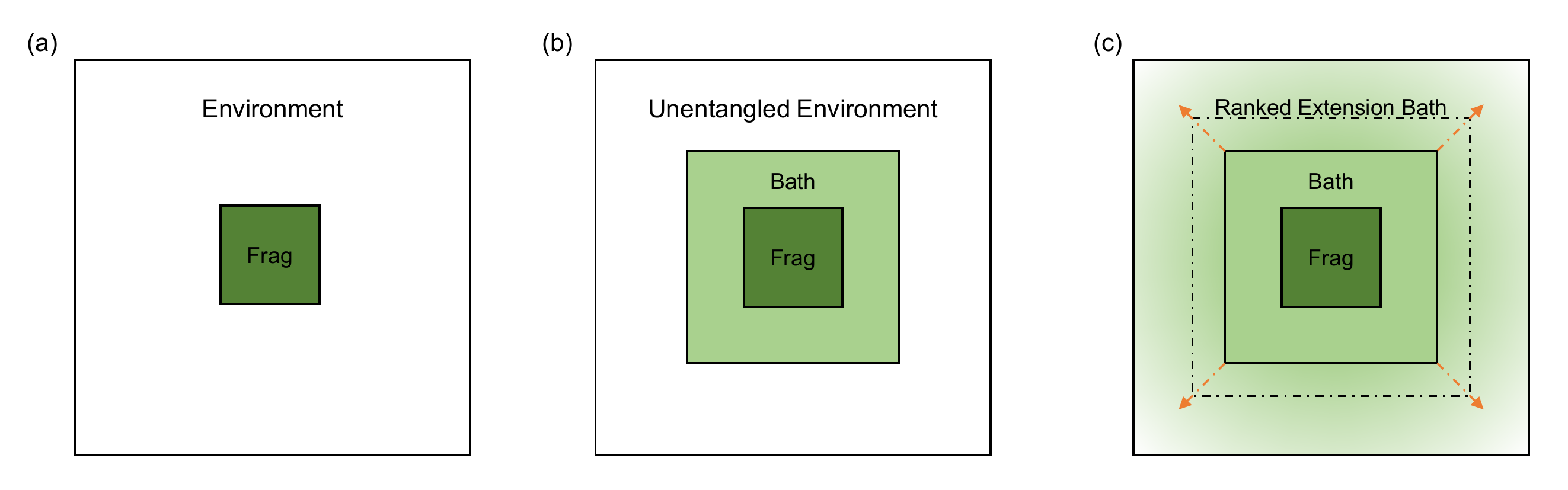}
\caption{Visualization of mentioned orbital sets in the Outline section. (a) There are $L_A$ LOs are selected into the fragment orbital set, and the rest of $L-L_A$ orbitals are named as environment set. General choices for fragment set are bonding orbitals or the orbitals in atom valence shell. (b) There are $L_B$ LOs with fractional occupied number are assigned into the bath orbital set. The fragment and bath set are termed as the Impurity. The $(L-L_A-L_B)$-scale unentangled-environment set contains core orbitals and virtural orbitals. (c) The ranked unentangled environment orbitals are gradually appended into the bath orbital set, until the bath set expands to the entire environment.}
\label{fig:vis}
\end{figure*}

The OE-VQE framework is demonstrated numerically for typical systems including hydrogen chain, hydrogen ring and nitrogen molecules, which are challenging strongly correlated systems in the dissociation distance. 
The numerical results show much higher accuracy for OE-VQE, which may only achievable with much deeper circuits and larger measurement complexity for other conventional quantum approaches. The proposed OE-VQE framework lies in the chemical insight that introduced as a novel dimension, paving the way for studying and analysis more advanced quantum eigensolvers.

\begin{figure*}
  \includegraphics[width=\textwidth]{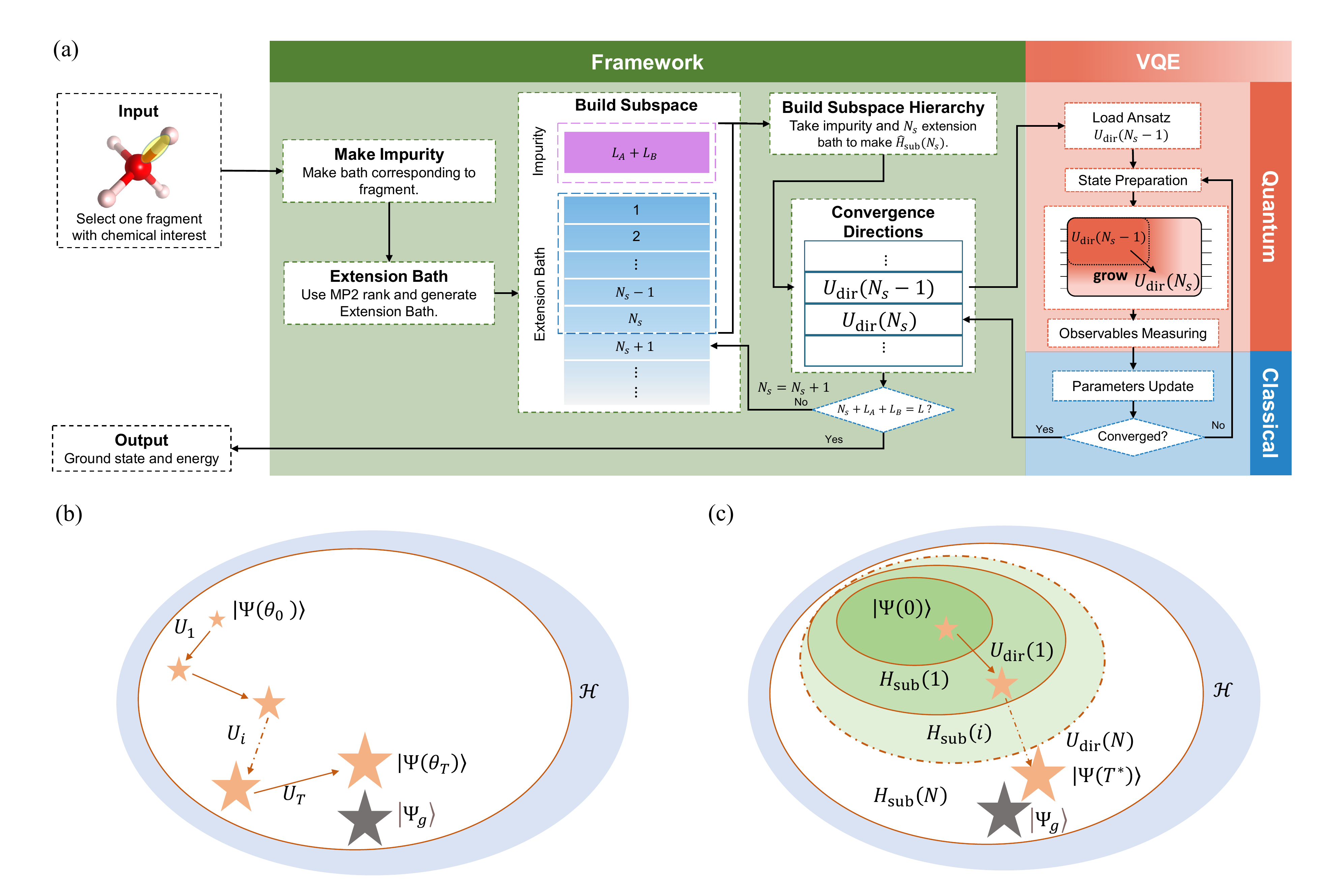}
  \caption{(a) Outline of the OE-VQE framework: For the outside loop (green panel), a systematic improvable convergence path is provided. The Impurity (fragment and bath) orbitals constructs a subspace $\hat{H}_{\rm sub}(0)$, and its ground state $|\Psi(0)\rangle$ is the starting point provided by OE-VQE framework. The environment orbitals are iteratively appended to the Impurity, and a subspace hierarchy $\{\hat{H}_{\rm sub}(N_s)\}_{N_s=1}^{N}$ is provided, where $N=L-L_A-L_B$. For the inside loop (red panel), each subspace determines a convergence path (quantum circuit) $U_{\rm dir}(N_s)$. The final stage of the OE-VQE framework will be an approximation to the ground state. (b) Visualization of the convergence path of a general quantum eigensolver without systematic improvability. Starting from $|\Psi(\theta_0)\rangle$, the quantum eigensolver ends up in $|\Psi(\theta_T)$ by following the path $U_1,U_2,...,U_T$. Here, $\mathcal{H}$ represents the whole Hilbert space. (c) Visualization of the convergence path of OE-VQE framework. The convergence direction is guaranteed by the subspace hierarchy.}\label{outline}
\end{figure*}

\section{Outline of the OE-VQE framework}
To clearly define the whole framework, several fundamental definitions and notations needed to be clarified before providing details. Here, we utilize $(p, q, r, s)$ to represent arbitrary Localized molecular Orbitals~(LO). LO is a kind of orthogonal orbital basis that inherits chemical characteristic from Molecular Orbital basis, meanwhile maintains a similar geometry structure to the Atomic Orbital basis. 
Then the electron Hamiltonian of a quantum chemical system under Born-Oppenheimer approximation can be formed as:
\begin{align}
    \hat{H}_e=E_{\rm nuc}+\sum\limits_{p,q}^Ld_{pq}\hat{a}_p^{\dagger}\hat{a}_q+\frac{1}{2}\sum\limits_{p,q,r,s}^Lh_{pqrs}\hat{a}_p^{\dagger}\hat{a}_q^{\dagger}\hat{a}_s\hat{a}_r,
\end{align}
where $E_{\rm nuc}$ is the nuclear repulsion energy, $d_{pq}$ ($h_{pqrs}$) represents single (double) electron integration, and $\hat{a}_p$ ($\hat{a}_p^{\dagger}$) denotes fermionic annihilation (creation) operator to the $p$-th orbital. 

The whole framework is composed of two fundamental phases, namely finding a \emph{starting point} and constructing convergence \emph{directions}.

\emph{Phase 1} starts from the density matrix $\bm D^{\rm HF}$ of Hartree-Fock state.  Firstly, $L_A$ LOs ($L_A \ll L$) are selected 
based on chemical insights. These $L_A$ LOs are usually named as \emph{fragment}
and the rest orbitals are named as \emph{environment}. The bonding orbitals or the orbitals in atom valence shell are common choices for the fragment part. Then delete these $L_A$ LOs from $\bm D^{\rm HF}$ and diagonalize the rest of $\bm D^{\rm HF}$ to select $L_B$ orbitals with fractional-occupied particle number. These $L_B$ orbitals have the most significant correlations to the fragment part, which are usually named as \emph{bath}, and these selected $L_A+L_B$ orbitals are named as \emph{Impurity}. 
The remainder integral-occupied and unoccupied orbitals in environment, termed as \emph{core} and \emph{virtual} respectively, are unentangled with the impurity at the Hartree-Fock theory level. Therefore, core and virtual orbitals compose the \emph{unentangled environment}. The impurity orbitals form an initial active space $\hat{H}_{\rm sub}(0)$, and its ground state $|\Psi(0)\rangle$ is the starting point of the OE-VQE. 

\emph{Phase 2} focuses on constructing convergence directions (quantum circuits) $\{U_{\rm dir}(N_s)\}_{N_s=1}^{L-L_A-L_B}$ by leveraging the rest of $(L-L_A-L_B)$ unentangled-environment with MP2 experience. Here, the systematical improvable directions are guided by a \emph{subspace hierarchy} $\{\hat{H}_{\rm sub}(N_s)\}_{N_s=1}^{L-L_A-L_B}$, where each subspace $\hat{H}_{\rm sub}(N_s)$ is determined by the particle number exchange between the impurity and environment. Specifically, a Hartree-Fock method implements on the impurity orbitals to provide $L_{\rm occ}$ occupied and $(L_A+L_B-L_{\rm occ})$ unoccupied orbitals information. After that, perform the MP2 method on impurity occupied and virtual in environment
(impurity unoccupied and core in environment) subspaces, and the $(L-L_A-L_B)$ unentangled environment orbitals will be naturally ranked based on the particle change. Impurity combined with the first-$N_s$ unentangled environment orbitals formulate a $(L_A+L_B+N_s)$-dimensional quantum subspace $\hat{H}_{\rm sub}(N_s)$ which gradually convergences to $\hat{H}_e$ with the increase of $N_s$. The visualizations of fragment, bath, environment, unentangled-environment and ranked extension bath orbitals are illustrated as Fig.~\ref{fig:vis}.

Denote $N=L-L_A-L_B$, for the number of orbitals $N_s\in[1, N]$, the proposed OE-VQE of electron Hamiltonians $\hat{H}_e$ can be summarized into following steps:
\begin{enumerate}
    \item Construct the subspace $\hat{H}_{\rm sub}(N_s)$;
    \item Initialize the quantum state 
    \begin{align}
    |\Psi_0(N_s)\rangle=\prod\limits_{m=1}^{L_{\rm occ}}\hat{a}_m^{\dagger}|0^{L_A+L_B}\rangle\otimes\prod\limits_{l=1}^{N_s}\hat{o}_{l}|0^{N_s}\rangle.
    \end{align}
If $l\in {\rm core}$, the operator $\hat{o}_l=\hat{a}_l^{\dagger}$, else  $\hat{o}_l=I_l$.
    \item Generate $|\Psi(N_s-1)\rangle=U_{\rm dir}(N_s-1)|\Psi_0(N_s)\rangle$ as the starting point in the $N_s$-th iteration. Here the quantum circuit $U_{\rm dir}(N_s-1)$ contains parameters and corresponding operators in the $(N_s-1)$-th iteration;
    \item Perform VQE programs on quantum device with Hamiltonian $\hat{H}_{\rm sub}(N_s)$ and $|\Psi(N_s-1)\rangle$, and using the optimized fermionic operators and corresponding parameters to update the convergence direction $$U_{\rm dir}(N_s-1)\mapsto U_{\rm dir}(N_s)$$ and output the ground state of $\hat{H}_{\rm sub}(N_s)$;
    \item If $N_s \neq N$, go back to step 1 and set $N_s = N_s + 1$. 
\end{enumerate}
Finally, the ground state of $\hat{H}_e$ can be approximated by
\begin{align}
    |\Psi_g\rangle=U_{\rm dir}(N)|\Psi_0(N)\rangle.
\end{align}
The whole procedure is visualized as Fig.~\ref{outline}.

\section{Technical Details of OE-VQE}
Here, elaborate technical details of OE-VQE framework are provided. We first introduce the method in constructing the starting point, that is the ground state of $\hat{H}_{\rm sub}(0)$. Then we introduce the workflow in 
showing how to efficiently build a convergence path for quantum eigensolvers.

\subsection{Construct the Starting Point}
This starts from the Hartree-Fock state
\begin{align}
    |\Phi_e\rangle=\prod\limits_{\mu\in N_{\rm occ}}\hat{a}_{\mu}^{\dagger}|0^{\otimes L}\rangle,
\end{align}
which is a general reference of ground state for molecular systems. Here $N_{\rm occ}$ represents occupied molecular orbital set and $\{\hat{a}_{\mu}^{\dagger},\mu\in[|N_{\rm occ}|]\}$ represent creation operators. Suppose a group of orthogonal Localized-Orbital (LO) basis $\{\hat{a}_k^{\dagger},k\in[L]\}$ has been selected via using Intrinsic Atomic Orbital (IAO) methods~\cite{knizia2013intrinsic}, a $L\times L$ reduced density matrix 
\begin{align}
    \bm D_{kl}^{\rm HF}=\langle\Phi_e|\hat{a}_k^{\dagger}\hat{a}_l|\Phi_e\rangle
\end{align}
is obtained, in which $(k,l)$ represent orbitals index in LO basis. Suppose there are $L_A$ LOs are selected to compose the fragment subspace. To do this, the $\bm D^{\rm HF}$ should be reorganized by moving each row and column in form of
\begin{align}
	     \bm D^{\rm HF}=\left(\begin{array}{cc}
	        \bm D^{A}_{L_A\times L_A} & \bm D^{\rm inter}_{L_A\times (L-L_A)} \\
	        \bm D^{\rm inter\dagger}_{L_A\times (L-L_A)} &   \bm D^{B}_{(L-L_A)\times (L-L_A)}\\
	     \end{array}
	     \right),
\end{align}
where $\bm D^{A}_{L_A\times L_A}$ represents to the fragment subspace and $\bm D^{\rm B}_{(L-L_A)\times (L-L_A)}$ represents to the environment. 
To fully decouple the fragment $\bm D^{A}_{L_A\times L_A}$ with the environment, it is necessary to analyze the entanglement distribution in environment orbitals. Specifically, the eigenvalues of density matrix $\bm D^{B}_{(L-L_A)\times (L-L_A)}$ provide the particle number in the environment, that is
\begin{align}
    \bm D^{B}_{(L-L_A)\times (L-L_A)}=\sum\limits_{p=1}^{L-L_A}\lambda_p|B_p\rangle\langle B_p|=\bm U^{B}\bm\lambda^B\bm U^{B\dagger}.
\end{align}
Here the diagonal matrix $\bm\lambda^B=\sum_{p=1}^{L-L_A}\lambda_p|p\rangle\langle p|$, $\lambda_p\in[0, 2]$ represents the particle number on a new environment orbital basis $|B_p\rangle$ transformed from LO basis by the unitary operator $\bm U^{B}$. In our work, the eigenvector $|B_p\rangle$ is termed as Unentangled Environment Orbital~(UEO) basis. According to the distribution of $\lambda_p$, the environment orbitals can be separated into three segments, including $L_B$ bath orbitals, $L_{\rm core}$ core orbitals and $L_{\rm vir}$ virtual orbitals. Obviously, the relationship $$L=L_A+L_B+L_{\rm core}+L_{\rm vir}$$ holds. The bath orbitals entangled with the fragment will contribute all the particles between 0 and 2, while core and virtual orbitals contain 2 particles and 0 particles, respectively. Due to MacDonald's theorem~\cite{macdonald1933successive}, the relationship $L_A\geq L_B$ holds, and a $L\times L$ coefficient transformation matrix can be formulated as
\begin{align}
    \bm U^{\rm LO\mapsto UEO}=\left(\begin{array}{cc}
	        I^A_{L_A\times L_A} & \bm 0_{L_A\times (L-L_A)} \\
	        \bm 0_{L_A\times (L-L_A)} &   \bm U^{B}\\
	     \end{array}
	     \right),
\end{align}
where the last $(L-L_A)$ columns can be decomposed as
\begin{align}
	     \bm U^{\rm bath(LO\mapsto UEO)}_{L\times L_B}\oplus \bm U^{\rm core(LO\mapsto UEO)}_{L\times L_{\rm core}}\oplus \bm U^{\rm vir(LO\mapsto UEO)}_{L\times L_{\rm virt}},
\end{align}
and the first $L_A$ columns are denoted as
\begin{align}
    \left(\begin{array}{c}
	        I^A_{L_A\times L_A}\\
	        \bm 0_{L_A\times (L-L_A)}\\
	     \end{array}
	     \right)=\bm U^{\rm frag(LO\mapsto UEO)}_{L\times L_A}.
\end{align}
Therefore, the impurity part is completely decoupled with core and virtual orbitals in the Hartree-Fock level. Denote the unitary matrix
\begin{align}
    \bm U^{\rm imp (LO\mapsto UEO)}_{L\times (L_A+L_B)}=\bm U^{\rm frag(LO\mapsto UEO)}_{L\times L_A}\oplus \bm U^{\rm bath(LO\mapsto UEO)}_{L\times L_B},
\end{align}
then the initial subspace can be constructed by
\begin{align}
    \hat{H}_{\rm sub}(0)=\bm U^{\rm imp(LO\mapsto UEO)\dagger}_{L\times (L_A+L_B)}\hat{H}_e\bm U^{\rm imp(LO\mapsto UEO)}_{L\times (L_A+L_B)},
    \label{Eq:embedding}
\end{align}
and the starting point $|\Psi(0)\rangle$ is the ground state of $\hat{H}_{\rm sub}(0)$. Details in calculating Eq.~\ref{Eq:embedding} are provided in Appendix.~\ref{App-C}.

\subsection{Construct Convergence Directions}
The directions are a series of quantum circuits $\{U_{\rm dir}(N_s)\}_{N_s=1}^{N}$ which are guided by the subspace hierarchy $\{\hat{H}_{\rm sub}(N_s)\}_{N_s=1}^{N}$, where $N=(L-L_A-L_B)$. The subspace hierarchy is described by impurity with $N_s$ extension bath orbitals. Here,  extension bath orbitals are ranked by their contributions to the particle number exchanging in occupied and unoccupied orbitals between the impurity and unentangled environment. Therefore the high-effective MP2 is utilized to calculate the particle number exchanging, that is
\begin{align}
    \delta\bm\lambda^{\rm env}=\left(\lambda_1^{\rm vir},..., \lambda^{\rm vir}_{L_{\rm vir}},2-\lambda_1^{\rm core},...,2-\lambda^{\rm core}_{L_{\rm core}}\right),
\end{align}
where $L_{\rm vir}+L_{\rm core}=N$, $\lambda_{i}^{\rm vir}$ and $2-\lambda_j^{\rm core}$ represent particle number exchange from environment virtual orbitals (core orbitals) to impurity occupied orbitals (unoccupied orbitals), respectively. 
Rank the particle number changes (from the largest to the smallest), these $N$ orbitals are reordered according to their entanglement contributions to the impurity, and the corresponding coefficient matrix
\begin{align}
    \bm U^{\rm LO\mapsto MUEO}=\left(\bm U^{\rm imp(LO\mapsto MUEO)}_{L\times (L_A+L_B)}, \bm U^{\rm env (LO\mapsto MUEO)}_{L\times N}\right)
\end{align}
can be efficiently computed (details refer to Appendix.~B), where MP2 induced unentangled environment orbital basis~(MUEO) is given by MP2 method. Details in calculating $\delta\bm\lambda^{\rm env}$ and $\bm U^{\rm LO\mapsto MUEO}$ are provided in the Appendix~B.

For $N_{s}\in[0,N]$ environment orbitals are selected, the $L\times(L_A+L_B+N_s)$ projector $P(N_s)$ is the first $L_A+L_B+N_s$ columns of $\bm U^{\rm LO\mapsto MUEO}$:
\begin{align}
   P(N_s)=\left(\bm U^{\rm imp(LO\mapsto MUEO)}_{L\times (L_A+L_B)}, \bm U^{\rm env (LO\mapsto MUEO)}_{L\times N_s}\right),
\end{align}
and the subspace is thus constructed by 
\begin{align}
    \hat{H}_{\rm sub}(N_s)=P^{\dagger}(N_s)\hat{H}_eP(N_s).
\end{align}

The convergence direction $U_{\rm dir}(N_s)$ contians a series of Hermitian operators $\hat{\bm\tau}=(\hat{\tau_1},\hat{\tau_2},...)$ as well as their corresponding classical parameters $(\theta_1,\theta_2,...)$, which can produce a quantum trial state $|\Psi(N_s)\rangle$ in minimizing the energy function
\begin{align}
    E(N_s)=\langle\Psi(N_s)|\hat{H}_{\rm sub}(N_s)|\Psi(N_s)\rangle.
    \label{eq:energy}
\end{align}
The quantum trial state 
\begin{equation}
\begin{aligned}
    |\Psi(N_s)\rangle=U_{\rm dir}(N_s)U_{\rm dir}(N_s-1)|\Psi_0(N_s)\rangle,
\end{aligned}
\end{equation}
where
\begin{align}
    |\Psi_0(N_s)\rangle=\prod\limits_{m=1}^{L_{\rm occ}}\hat{a}_m^{\dagger}|0^{2L_A}\rangle\otimes\prod\limits_{l=1}^{N_s}\hat{o}_{l}|0^{N_s}\rangle
\end{align}
represents the lowest electron occupied state on the subspace $\hat{H}_{\rm sub}(N_s)$, the operator
\begin{align}
    \hat{o}_{l}=\left\{
    \begin{aligned}
    &\hat{a}_l^{\dagger}, \; {\rm if}\;  l\in{\rm core} \\
    &\hat{I}_l, \; \; {\rm if} \; l\in{\rm virt} 
    \end{aligned}
    \right.
\end{align}
only performs on the extension bath orbital, and $U_{\rm dir}(N_s-1)$ represents the trained direction (quantum circuit) in the $(N_s-1)$-th step. The essential idea in constructing $U_{\rm dir}(N_s)$ relies on selecting suitable Hermitian operators $\hat{\bm\tau}=(\hat{\tau}_1,\hat{\tau}_2,...)$ from the single- and double- excitation operator set
\begin{align}
    \left\{(\hat{a}_q^{\dagger}\hat{a}_p-\hat{a}_p^{\dagger}\hat{a}_q), (\hat{a}_p^{\dagger}\hat{a}_q^{\dagger}\hat{a}_r\hat{a}_s-\hat{a}_s^{\dagger}\hat{a}_r^{\dagger}\hat{a}_q\hat{a}_p)\right\}
\end{align}
where orbital indexes $(p,q,r,s)\in[L_A+L_B+N_s]$. Then a series of classical parameters $\bm\theta=(\theta_1,\theta_2,...)$ are obtained via minimizing energy function $E(N_s)$ in Eq.~\ref{eq:energy}.

Finally, the ground state energy of the molecule $\hat{H}_e$ can be calculated by
\begin{align}
    E_g = E_{\rm sub}(N_s) + E_{\rm core} + E_{\rm nuc},
\end{align}
where $E_{\rm sub}(N_s)=\min_{\bm\theta,\hat{\bm\tau}} E(N_s)$ is the ground state energy of the subspace $H_{\rm sub}(N_s)$, $E_{\rm core}$ is contributed by core orbitals in unentangled environment in MUEO basis, and $E_{\rm nuc}$ is the nuclear repulsion energy. Details refer to the Appendix.~\ref{energy_compute}.

\section{Enhance the performance of general quantum eigensolvers}
 Given an electron Hamiltonian $\hat{H}_e$, the method on selecting fermionic operators and parameters to construct a quantum ansatz have been widely studied \cite{o2016scalable, whitfield2011simulation, mcclean2016theory, barkoutsos2018quantum, grimsley2019adaptive, zhang2020mutual,tang2021qubit2,yordanov2021qubit}. We take a canonical ansatz as examples to show the OE-VQE framework can dramatically improve its performances.
\subsection{OE-ADAPT-VQE}
The essential idea of the OE-ADAPT-VQE is to approximate FCI by using a maximally compact sequence of operators in the operator pool $$\mathcal{P}(N_s)=P(L_A+L_B+N_s)-P(L_A+L_B+N_s-1),$$ 
where $P(L)$ contains single- and double- excitation operators on orbital indexes $[N]$. Each iteration starts from obtaining the parameter gradient by implementing quantum measurements to each $\hat{\tau}_i\in \mathcal{P}(N_s)$. Using Jordan-Wigner transformation, the fermionic operator $\hat{\tau}_i$ and $\hat{H}_{\rm sub}(N_s)$ can be encoded by Pauli operators, and the gradient information 
\begin{align}
    \frac{\partial E(N_s)}{\partial \theta}=\langle\Psi_{\rm init}(N_s)|e^{\theta\hat{\tau}^{\dagger}}\left[\hat{H}_{\rm sub}(N_s),\hat{\tau}\right]e^{\theta\hat{\tau}}|\Psi_{\rm init}(N_s)\rangle
\end{align}
can be efficiently estimated within an $\epsilon$-additive error by running Hardmard-Test algorithm $\mathcal{O}(\|\hat{H}_{\rm sub}(N_s)\|_1/\epsilon^2)$ times. The purpose of these gradient measurements is to determine the most suitable operator $\hat{\tau}_i\in \mathcal{P}(N_s)$ to grow the quantum ansatz at the current stage, as the operator with the largest gradient might have more contributions to the energy function. If the largest gradient information is larger than a threshold, the operator $\hat{\tau}_i$ will be appended to the quantum ansatz, and its corresponding parameter is obtained by minimizing the energy function
\begin{align}
    \theta_i^*=\arg \min\limits_{\theta_i}\langle\Psi_{\rm init}(N_s)|e^{\theta_i\hat{\tau}_i^{\dagger}}\hat{H}_{\rm sub}(N_s)e^{\theta_i\hat{\tau}_i}|\Psi_{\rm init}(N_s)\rangle,
\end{align}
where $|\Psi_{\rm init}(N_s)\rangle=U_{\rm dir}(N_s)|\Psi_0(N_s)\rangle$.
Repeat the above procedure several times, the selected operators and optimized parameters construct a new convergence direction $U_{\rm dir}(N_s)$. The whole procedure is summarized as Alg.~\ref{alg_si_adapt}.

\begin{algorithm}[H]
\caption{OE-ADAPT-VQE}\label{alg_si_adapt}
\begin{algorithmic}
\Require  subspace $\hat{H}_{\rm sub}(N_s)$, quantum measurement threshold $\epsilon$, direction $U_{\rm dir}(N_s-1)$
\Ensure New direction $U_{\rm dir}(N_s)$, energy $E_g$
\State Set $G_{\rm Vec}=\{\}$, $\Lambda_s=2\epsilon$, $U_{\rm dir}(N_s)=U_{\rm dir}(N_s-1)\otimes I^{\otimes 2}$
\While {$\Lambda_s>\epsilon$}
\For {operator $\hat{\tau}_i$ in $\mathcal{P}(N_s)$}
\State Calculate the gradient information $\frac{\partial E(N_s)}{\partial \theta_i}$;
\State $G_{\rm Vec}\cup\abs{\frac{\partial E(N_s)}{\partial \theta_i}}$
\EndFor
\State Select operators $\hat{\tau}^{*}=\arg\max\limits_{\tau}\left(G_{\rm Vec}\right)$, $\Lambda_s=\max\limits_{\tau}\left(G_{\rm Vec}\right)$
\State Grow ansatz $$|\Psi(N_s)\rangle=e^{0\hat{\tau}^{*}}U_{\rm dir}(N_s-1)|\Psi_0(N_s)\rangle$$
\State Minimize the energy $$E(N_s)=\min\limits_{\bm\theta}\langle\Psi(N_s)|\hat{H}_{\rm sub}(N_s)|\Psi(N_s)\rangle$$
\State Calculate the ground state energy of $\hat{H}_e$
$$E_g = E(N_s) + E_{\rm core} + E_{\rm nuc}$$
\State Update $\bm\theta$, $\tau^{*}$ in $U_{\rm dir}(N_s)$
\EndWhile
\noindent\Return {$U_{\rm dir}(N_s)$, $E_g$}

\end{algorithmic}
\end{algorithm}



\subsection{Measurement Complexity Reduction}
The barren plateaus phenomenon is ubiquitous in optimizing variational quantum algorithms if the information of the ansatz rapidly spreads to the whole Hilbert space, that is the gradient of the cost function would be exponentially vanishing with the increase of the number of qubit. Here, we report that the general ADAPT-VQE also suffers from barren plateaus phenomenon which dramatically reduces its efficiency.

For our general results, we consider a family of $L$-orbital electron Hamiltonian set $S=\{\hat{H}_e|{\rm Tr}(\hat{H}_e)\leq {\rm poly}(L)\}$. In the original ADAPT-VQE framework, the gradient of energy in terms of $\hat{\tau}_k\in P(L)$ can be calculated by
\begin{align}
    \frac{\partial E(\bm\theta)}{\partial\hat{\tau}_k}=\langle\Psi|\left[\hat{H}_e,\hat{\tau}_k\right]|\Psi\rangle,
\end{align}
where $\hat{H}_e\in S$, $|\Psi\rangle$ represents the reference state and $P(L)$ represents the fermionic operator pool. We show that such $\frac{\partial E(\bm\theta)}{\partial\hat{\tau}_k}$ will be exponentially vanishing in the average-case scenario.

\begin{theorem}
There exists a family of $L$-orbital electron Hamiltonian set $S=\{\hat{H}_e|{\rm Tr}(\hat{H}_e)\leq {\rm poly}(L)\}$, the gradient of ADAPT-VQEs suffers from barren plateaus phenomenon, that is
\begin{align}
    \mathbb{E}_{\hat{H}_e\sim S}\left[{\rm Var}_{\hat{\tau}_k\sim P(L)}\frac{\partial E(\bm\theta)}{\partial\hat{\tau}_k}\right]\leq \mathcal{O}\left(\frac{c}{4^L}\right),
\end{align}
where $E(\bm\theta)$ is the energy function of ADAPT-VQE, $P(L)$ is the operator pool and constant $c$ is independent to $L$. 
\label{theorem1}
\end{theorem}

The above theorem shows that the gradient fluctuates around the mean value with an exponentially small amplitude, and distinguish the largest gradient from the operator pool $P(L)$ thus requires exponentially amount of quantum measurement. In the following, we show that OE-VQE framework can relieve this phenomenon to some extent, and thus enhances practicability of popular heuristic quantum eigenslovers. 

\begin{corollary}
Given a $L$-orital electron Hamiltonian $\hat{H}_e$, suppose the OE-VQE chooses $L_A$ fragment orbitals and $L_B$ bath orbitals, and $M_{\rm OE-ADAPT}$ ($M_{\rm ADAPT}$) represent the quantum measurement complexity of OE-ADAPT-VQE (ADAPT-VQE) in approximating the ground state of $\hat{H}_e$, then the relationship
\begin{align}
   \frac{M_{\rm OE-ADAPT}}{M_{\rm ADAPT}}= \mathcal{O}\left(\frac{1}{(L-L_A-L_B)}\right).
\end{align}
holds.
\label{corollary1}
\end{corollary}
Here, $M_{\rm OE-ADAPT}$ ($M_{\rm ADAPT}$) represents the required measurement complexity in estimating gradient information in the whole VQE procedure.
The insight of Corollary~\ref{corollary1} relies on the orbitals with most correlations have been characterized by small-scale subspaces, and operators $\tau\in P(L)$ only have small contributions to the ground state energy. This reduces the required number of operators $\tau\in P(L)$ and results in less quantum measurement complexity. 

\begin{figure*}
  \begin{center}
  \includegraphics[width=\textwidth]{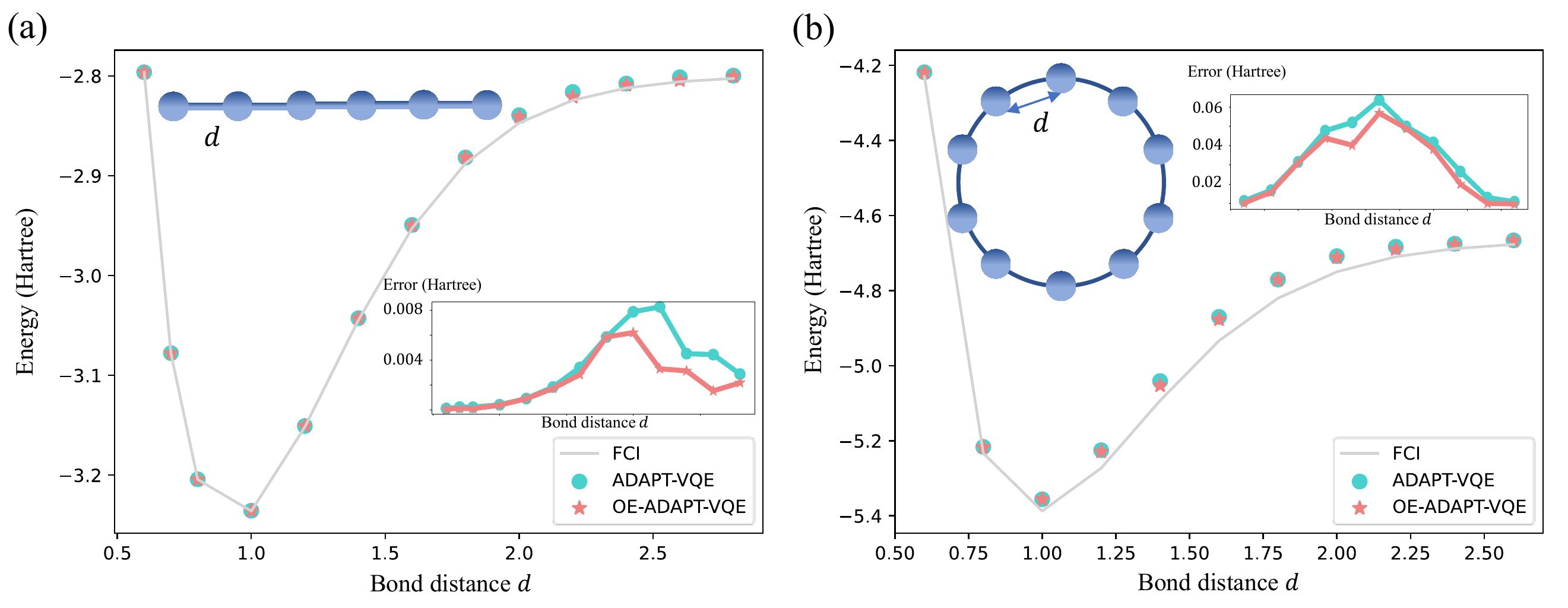}
  \caption{Numerical results to illustrate the ground state energy curve for (a) H$_6$ chain molecule and (b) H$_{10}$ ring molecule in sto-3g basis. Embedded subgraphs show the energy error by ADAPT-VQE algorithm and OE-ADAPT-VQE algorithm.}\label{fig:h6}
  \end{center}
\end{figure*}

\vspace{0.2cm}
\section{Simulation results}
In this section, we explore the performance of OE-VQE framework with a few molecular systems, including $\rm H_6$ chain, $\rm H_{10}$ ring and $\rm N_2$ in sto-3g basis. The presented molecules are prototypical strongly correlated systems in the dissociation scenario which can not be accurately described by previous quantum methods, and OE-VQE framework shows better performance on such strongly correlated systems. Here, we utilize Pyscf~\cite{sun2018pyscf} for the molecular integral calculation and HF method calculation, and the FQE~\cite{rubin2021fermionic} is used to simulate Fermionic operator calculation. All iteration steps utilized the Broyden-Fletcher-Goldfarb-Shannon~(BFGS) algorithm within Scipy to tune parameters~\cite{virtanen2020scipy}.

\subsection{The Hydrogen chain and Hydrogen ring}
While CCSD methods are exceptionally accurate for equilibrium geometries, it often fails for out-of-equilibrium geometries such as molecules with low-lying excited states. As typical cases, understanding Hydrogen chain and Hydrogen ring is critical for predicting many chemical properties. Here, we consider the ground states of $\rm H_6$ chain and $\rm H_{10}$ ring molecules which are believed to be strongly correlated systems in the dissociation distance. The $\rm H_6$ chain molecule and H$_{10}$-ring molecule are encoded as qubit Hamiltonians with $12$ qubits and $20$ qubits, respectively. To show the performance of OE-VQE framework, we only append $100$ operators to OE-ADAPT-VQE and ADAPT-VQE to test their performances. Fig.~\ref{fig:h6} (a) and (b) show that OE-ADAPT-VQE provides higher accuracy compared with ADAPT-VQE in both Hydrogen chain and ring models, especially in the dissociation domain ($d\in [2.0\mathring{{\rm A}}, 2.6\mathring{{\rm A}}]$).
Compared with $\rm H_6$ chain, $\rm H_{10}$ ring is a much more difficult molecule, since its ring-like geometry structure enables each atom sharing same chemical properties~\cite{motta2020ground}, which seriously troubles the selection of fragment orbitals based on chemical insights. Even in this scenario, OE-VQE still obtains a higher accuracy compared with original ADAPT-VQE by only using $100$ operators and less measurement complexity. This also explains relatively larger error in Fig.~\ref{fig:h6} (b). Simulation results clearly show that the OE-VQE framework can capture more valuable correlations from the active space which provides a better convergence path to the ADAPT-VQE method.

\begin{figure*}
  \includegraphics[width=0.95\textwidth]{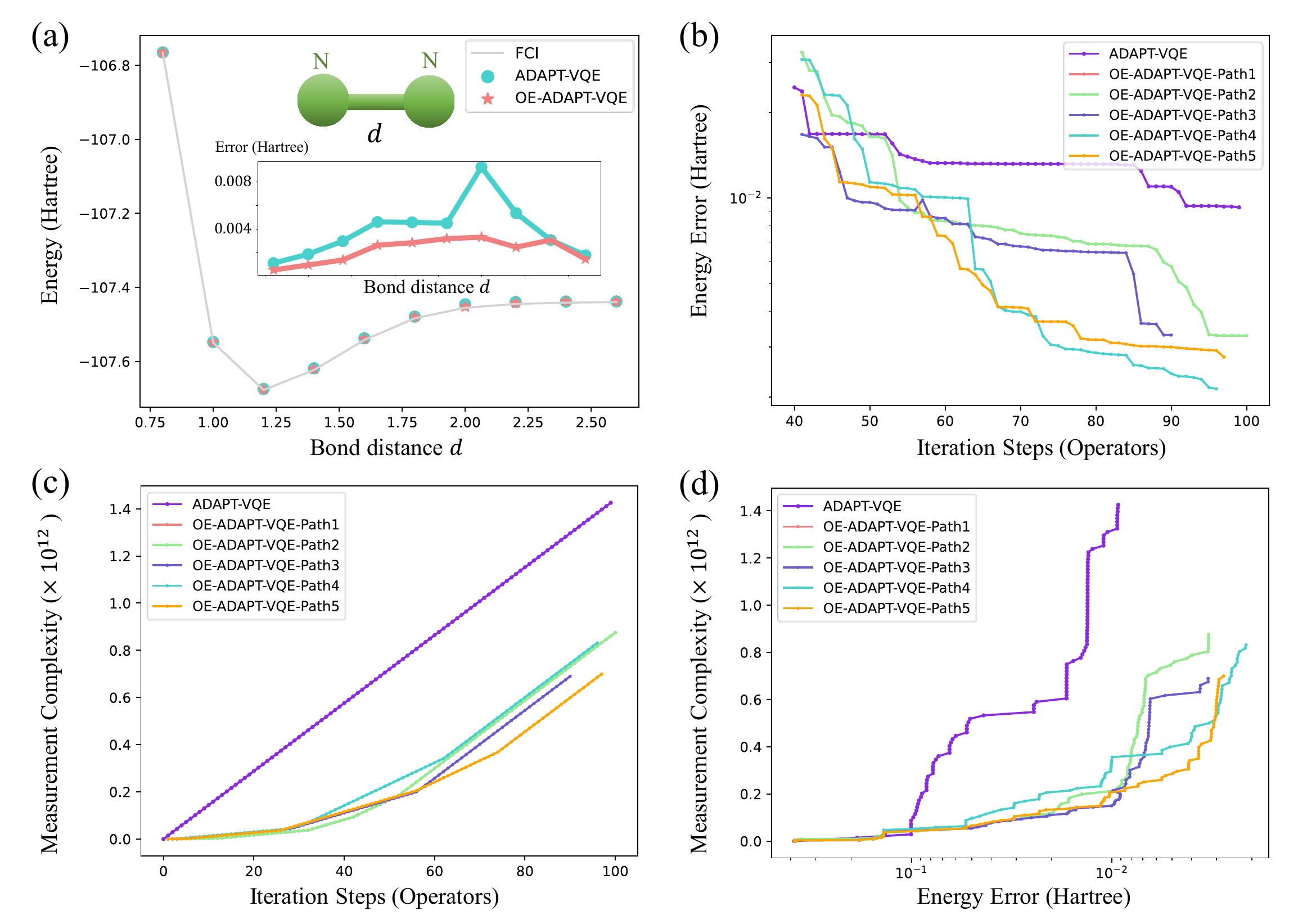}
  \caption{Numerical results on Nitrogen molecule in sto-3g basis. (a) Ground state energy curve provided by OE-ADAPT-VQE and ADAPT-VQE in the range $d\in[0.8, 2.6]$. (b) The relationship between energy error and iteration steps. With the increase of Fermionic operators, OE-ADAPT-VQE convergence to a better solution more rapidly. (c) The rising tendency of required measurement complexity with the increase of iteration steps. (d) Required measurement complexity on different energy error.}\label{fig:N2}
\end{figure*}

\subsection{The Nitrogen molecule}
The Nitrogen molecule is a strong correlation system that troubles coupled cluster methods, where its triple bond is hard to describe accurately and the stretched configuration has several low-lying excited states. These features lead to errors when using single-reference coupled cluster methods~\cite{lyakh2012multireference, pfau2020ab}. Fig.~\ref{fig:N2} provides comparisons between OE-ADAPT-VQE and original ADAPT-VQE in terms of energy error, circuit depth (required operators) and measurement complexity. In the range of $d\in[0.8\mathring{{\rm A}}, 2.6\mathring{{\rm A}}]$, OE-ADAPT-VQE shows much higher accuracy compared with ADAPT-VQE, meanwhile the former method only requires shallower depth circuit and less measurement complexity. We show OE-ADAPT-VQE achieves the same accuracy to ADAPT-VQE with much shallow circuit depth (used fermionic operators) in Table~\ref{tab:circuit_depth}. For the Nitrogen molecule, $2$ bonding orbitals are selected as the impurity, and the rest $8$ spatial orbitals are iteratively appended to the initial subspace. Different energy convergence conditions in each subspace form various convergence paths, and
we test five different convergence paths provided by OE-ADAPT-VQE. Fig.~\ref{fig:N2} (b) shows that systematical improbable paths are much more powerful compared with single gradient-based strategy. Fig.~\ref{fig:N2} (b) and (c) illustrate the OE-ADAPT-VQE dramatically reduces the quantum measurement complexity but achieves much higher accuracy.
In this case, the measurement complexity counts quantum shots in estimating gradient information during appending these $100$ fermionic operators, that is $M_{\rm Adapt}=100\|P(10)\|\epsilon^{-2}$, and $M_{\rm OE-Adapt}=\sum_{i=2}^{10}n_i\|P(i)\|\epsilon^{-2}$, where measurement accuracy $\epsilon=10^{-3}$, $\|P(i)\|$ represents the scale of operator pool with $i$ spatial orbitals, and $n_i$ represents the number of selected operators in the $i$-th subspace. This shows that the OE-VQE framework is a clear improvement over the original ADAPT-VQE for modeling a strongly correlated chemical system. 

\begin{table}[t]
    \centering
    \setlength{\tabcolsep}{3mm}{
    \begin{tabular}{c c c c}
    \hline\hline
        $d$ $(\mathring{{\rm A}})$ & Error & \makecell[c]{OE-ADAPT-VQE \\ (Num of Op)}  &  \makecell[c]{ADAPT~\cite{grimsley2019adaptive, rubin2021fermionic} \\ (Num of Op)}\\
        \hline
        0.8 &   0.00107   &     75      &   100\\
        \hline
        1.0 &   0.00183   &     67      &   100\\
        \hline
        1.2  &  0.00295   &     64      &   100\\
        \hline
        1.4 &   0.00459   &     64      &   100\\
        \hline
        1.6 &   0.00455   &     64      &   100\\
        \hline
        1.8 &   0.00447   &     64      &   100\\
        \hline
        2.0 &   0.00928   &     57      &   100\\
        \hline
        2.2 &   0.00534   &     57      &   100\\
        \hline
        2.4 &   0.00304   &     80      &   100\\
        \hline
        2.6 &   0.00174   &     37      &   100\\
        \hline\hline
    \end{tabular}}
    \caption{Utilized Fermionic Operators comparison in achieving the same energy error for OE-ADAPT-VQE and ADAPT-VQE.} 
    \label{tab:circuit_depth}
\end{table}

\section{Discussion}
With the development of the quantum hardware, there will be more widely interests to demonstrate their computational power for solving complex problems, one central concern is to generate a powerful ansatz with shallow quantum circuit due to the inevitable quantum noise and limited coherence time. In this work, the proposed OE-VQE framework builds an efficient convergence path in solving ground state problems, results in a powerful ansatz with low-depth quantum circuit. To investigate the accuracy of our framework, we studied ground state energies of hydrogen chain, hydrogen ring and nitrogen molecules in the dissociation scenario, which are challenge problems for (unitary) coupled cluster methods. Our numerical results show much higher accuracy compared with the typical quantum eigensolver ADAPT-VQE.

We also analyze the measurement complexity of heuristic quantum eigenslovers. As a particular example, we rigorously proved that the ADAPT-VQE suffers from the barren plateaus phenomenon, results in the energy will be end up in a flat region of the landscape with large probability. This phenomenon suggests that heuristic quantum solvers may require a large amount of measurement complexity when constructing a powerful ansatz. Our analysis shows the OE-VQE framework can relieve this phenomenon to some extent, and thus enhances practicability of popular heuristic quantum eigenslovers.

Meanwhile, the OE-VQE framework leaves room for further explorations. In this work, the entanglement-oriented convergence path is constructed based on the particle exchange number, which is one of the promising methods to rank unentangled environment orbitals. Therefore, further investigation needs to be conducted on testing other possible orbital rank criteria in the OE-VQE framework. Hopefully, there will exist other more advanced criteria in achieving better performance. Another possible perspective is to use post Hartree-Fock methods in a self-consistent way: in each bath extension procedure, the rotated and reformulated orbitals will be revised by the output result from the previous step. Furthermore, find a novel method to compress more correlations in impurity with the help of quantum computation would be interested in practice. 

Finally, we believe that recent quantum technical advances can be improved in both accuracy and efficiency by introducing our OE-VQE framework. One of the promising directions relates to the Density Matrix Embedding Theory~(DMET) method, which intrinsically allows high-level treatments for multiple fragments computation at the same time. The OE-VQE framework may help DMET method solving large-scale molecules and materials with less quantum resources. Another promising direction may employ the OE-VQE framework as the quantum-classical hybrid quantum Monte Carlo~(QMC) solver, for example, a quantum-classical hybrid full configuration interaction QMC~\cite{zhang2022quantum} or auxiliary field QMC~\cite{huggins2022unbiasing}. The performance of these QMC enhanced methods could be influenced dramatically by the initial trail state construct from quantum computers.



\section{acknowledgement}
The authors would like to thank Changsu Cao, Yifei Huang and Yanqi Song to provide helpful suggestions for the manuscript and Hang Li for support and guidance. Y. Wu is supported by the China Scholarship Council (Grant No.~202006470011).

\bibliography{main.bbl}

\clearpage
\widetext

\appendix

\section{Background Review}
We first consider a warm-up scenario which constructs subspace from full configuration interaction (FCI) level. Suppose the Hamiltonian $\hat{H}_e$ is composed of two parts, a central fragment $A$ with $L_A$ orbitals and an environment $B$ with $L_B$ orbitals. In general, any eigenstate of $\hat{H}_e$ can be expressed by the superposition in $\{|A_i\rangle\otimes|B_j\rangle\}$ with dimension of $L_A\times L_B$, where $|A_i\rangle\in A$ and $|B_j\rangle\in B$:
\begin{align}
    |\Psi_e\rangle=\sum\limits_{i}^{L_A}\sum\limits_j^{L_B}\Psi_{ij}|A_i\rangle|B_j\rangle,
\end{align}
where $\Psi_{ij}$ are complex matrices. Considering to utilize singular value decomposition $\Psi_{ij}=U_{i\alpha}\lambda_{\alpha}V_{\alpha j}^{\dagger}$, $|\Psi_e\rangle$ can be further reduced into
\begin{align}
    |\Psi_e\rangle=\sum\limits_{\alpha}^{L_A}\lambda_{\alpha}|\tilde{A}_{\alpha}\rangle|\tilde{B}_{\alpha}\rangle,
\end{align}
in which $|\tilde{B}_{\alpha}\rangle$ are defined as exact bath orbitals for the fragment space $A$. Under this expression, if $|\Psi_e\rangle$ denotes the ground state of $H_e$ in the space $\mathcal{H}^A\otimes\mathcal{H}^B$, then it is also the ground state of the subspace $\hat{H}_{\rm sub}=P\hat{H}_eP$, where the projector matrix
\begin{align}
    P=\sum\limits_{\alpha,\beta}^{L_A}|\tilde{A}_{\alpha}\tilde{B}_{\beta}\rangle\langle\tilde{A}_{\alpha}\tilde{B}_{\beta}|.
\end{align}
In general, the above construction depends on the exact ground state $|\Psi_e\rangle$ of the entire molecule system $\hat{H}_e$ which is an unrealistic approach in practice. In the main manuscript, we thus consider to match the density matrix of the subspace $\hat{H}_{\rm sub}$ and the entire system $\hat{H}_e$ at the Hartree-Fock level.

\section{Subspace Hierarchy Construction}
Here, we provide technical details of constructing subspace hierarchy. The whole procedure is summarized as Alg.~\ref{alg_subspace_construction}.
\begin{algorithm}[H]
\caption{Subspace Construction}\label{alg_subspace_construction}
\begin{algorithmic}
\Require molecular Hamiltonian $\hat{H}_e$, Slater Determinant $|\Phi_e\rangle$, $L_A$ fragment orbitals, and subspace extension index $N_s$
\Ensure Subspace $\hat{H}_{\rm sub}(N_s)$
\State 1. Calculate density matrix $\bm D^{\rm LO}$ by $|\Phi_e\rangle$;
\State 2. Based on the $L_A$ fragment orbitals, decompose the coefficient transformation matrix $\bm U^{\rm LO\mapsto UEO}$ into 
$$\left(\bm U^{\rm frag(LO\mapsto UEO)}_{L\times L_A},\bm U^{\rm bath(LO\mapsto UEO)}_{L\times L_B},\bm U^{\rm core(LO\mapsto UEO)}_{L\times L_{\rm core}},\bm U^{\rm vir(LO\mapsto UEO)}_{L\times L_{\rm virt}}\right);$$
\State 3. Map fragment and bath orbitals to $\rm UEO$ basis to acquire occupied and unoccupied information:
$$\hat{F}^{\rm UEO}=\left(\bm U^{\rm imp (LO \mapsto UEO)}\right)^{\dagger}\hat{F}^{\rm LO}\bm U^{\rm imp (LO \mapsto UEO)};$$
\State 4. Perform MP2 method on subspaces $\bm U^{\rm occ(LO\mapsto MUEO)}_{L\times L_{\rm occ}}\oplus\bm U^{\rm vir(LO\mapsto UEO)}_{L\times L_{\rm vir}}$ and  $\bm U^{\rm unocc(LO\mapsto MUEO)}_{L\times L_{\rm unocc}}\oplus\bm U^{\rm core(LO\mapsto UEO)}_{L\times L_{\rm core}}$ respectively, and obtain particle number change:
$$\delta\bm\lambda^{\rm env}=\left(\lambda_1^{\rm vir},..., \lambda^{\rm vir}_{L_{\rm vir}},2-\lambda_1^{\rm core},...,2-\lambda^{\rm core}_{L_{\rm core}}\right);$$
\State 5. Using $\delta\bm\lambda^{\rm env}$ to calculate
$$\bm U^{\rm LO\mapsto MUEO}=\left(\bm U^{\rm imp(LO\mapsto MUEO)}_{L\times (L_A+L_B)}, \bm U^{\rm env (LO\mapsto MUEO)}_{L\times(L_{\rm core}+L_{\rm vir})}\right);$$
\State 6. Calculate projector $P(N_s)$;\\
\noindent\Return {$\hat{H}_{\rm sub}(N_s)=P(N_s)\hat{H}_eP^{\dagger}(N_s)$}
\end{algorithmic}
\end{algorithm}
Specifically, the occupied and unoccupied orbitals in the impurity (LO basis) can be calculated by the Fock matrix
\begin{align}
\label{eq:Fock_LO}
    \hat{F}^{\rm LO}_{ij}=d_{ij}+\sum\limits_{kl}^LD_{kl}^{\rm HF}\left(\langle ij|lk\rangle-\frac{1}{2}\langle ik|lj\rangle\right),
\end{align}
where $\langle ij|lk\rangle$ is the two body term coefficient in Hamiltonian and equal to $h_{ijkl}$. Since the fragment and bath orbitals have been embedded into the UEO basis, it is natural to transform $\hat{F}^{\rm LO}$ to the UEO basis by using the linear map from LO to UEO:
\begin{align}
    \bm U^{\rm imp (LO \mapsto UEO)}=\left(
	       \bm U_{L\times L_A}^{\rm frag (LO\mapsto UEO)}, \bm U_{L\times L_B}^{\rm bath (LO\mapsto UEO)}
	     \right).
\end{align}
Then the Fock matrix in the UEO basis can be calculated by
\begin{align}
   \hat{F}^{\rm UEO}=\left(\bm U^{\rm imp (LO \mapsto UEO)}\right)^{\dagger}\hat{F}^{\rm LO}\bm U^{\rm imp (LO \mapsto UEO)}.
\end{align}
The spectral information of $\hat{F}^{\rm UEO}$ provides the occupied and unoccupied information of the impurity, that is
\begin{align}
    \hat{F}^{\rm UEO}=\bm U^{\rm imp(UEO\mapsto MUEO)\dagger}\left(\sum\limits_{j=1}^{L_A+L_B}\lambda^{I}_j|j\rangle\langle j|\right)\bm U^{\rm imp(UEO\mapsto MUEO)}
\end{align}
where eigenvalues $(\lambda^I_1,...\lambda^I_{L_A+L_B})$ represent energy contributions for each orbital. 

Rank these orbitals according to the energies (from the smallest to the largest), the first $[l/2]$ orbitals are assigned to occupied and the rest orbitals are assigned to unoccupied, where $l$ represents the electron number in the impurity. And their corresponding coefficient matrix can be obtained by
\begin{equation}
\begin{aligned}
    \bm U^{\rm imp(LO\mapsto MUEO)}_{L\times (L_A+L_B)}=\bm U^{\rm imp(LO\mapsto UEO)}_{L\times (L_A+L_B)}\bm U^{\rm imp(UEO\mapsto MUEO)}_{L\times (L_A+L_B)}
    =\left(\bm U^{\rm occ(LO\mapsto MUEO)}_{L\times L_{\rm occ}}, \bm U^{\rm unocc(LO\mapsto MUEO)}_{L\times L_{\rm unocc}}\right).
\end{aligned}
\end{equation}
After that, implement MP2 method in subspaces $\bm U^{\rm occ(LO\mapsto MUEO)}_{L\times L_{\rm occ}}\oplus\bm U^{\rm vir(LO\mapsto UEO)}_{L\times L_{\rm vir}}$ and  $\bm U^{\rm unocc(LO\mapsto MUEO)}_{L\times L_{\rm unocc}}\oplus\bm U^{\rm core(LO\mapsto UEO)}_{L\times L_{\rm core}}$ respectively, then one can analyze the electron excitation from occupied (core) orbitals to virtual (unoccupied) orbitals. 

Utilizing the MP2 method on the above two subspaces, and MP2 method outputs their density matrices
\begin{align}
    \bm D^{\rm occ, vir}=\left(\begin{array}{ll}
	        \bm D^{\rm occ}_{L_{\rm occ}\times L_{\rm occ}} & \bm 0_{L_{\rm occ}\times L_{\rm vir}} \\
	        \bm 0_{L_{\rm vir}\times L_{\rm occ}} &   \bm D^{\rm vir}_{L_{\rm vir}\times L_{\rm vir}}
	     \end{array}
	 \right),
\end{align}
and
\begin{align}
    \bm D^{\rm core, unocc}=\left(\begin{array}{cc}
	        \bm D^{\rm core}_{L_{\rm core}\times L_{\rm core}} & \bm 0_{L_{\rm core}\times L_{\rm unocc}} \\
	        \bm 0_{L_{\rm unocc}\times L_{\rm core}} &   \bm D^{\rm unocc}_{L_{\rm unocc}\times L_{\rm unocc}}\\
	     \end{array}
	     \right).
\end{align}
Since density matrics $\bm D^{\rm core}_{L_{\rm core}\times L_{\rm core}}$ and $\bm D^{\rm vir}_{L_{\rm vir}\times L_{\rm vir}}$ record the electron excitation which can be witnessed by density matrix diagnalization
\begin{align}
    \bm D^{\rm core}_{L_{\rm core}\times L_{\rm core}}=\bm U^{\rm core(UEO\mapsto MUEO)}_{L_{\rm core}\times L_{\rm core}}\bm\lambda^{\rm core}\bm U^{\dagger \rm core(UEO\mapsto MUEO)}_{L_{\rm core}\times L_{\rm core}},
\end{align}
as well as
\begin{align}
     \bm D^{\rm vir}_{L_{\rm vir}\times L_{\rm vir}}=\bm U^{\rm vir(UEO\mapsto MUEO)}_{L_{\rm vir}\times L_{\rm vir}}\bm\lambda^{\rm vir}\bm U^{\dagger \rm vir(UEO\mapsto MUEO)}_{L_{\rm vir}\times L_{\rm vir}},
\end{align}
then particle numbers exchange is obtained:
\begin{align}
    \delta\bm\lambda^{\rm env}=\left(\lambda_1^{\rm vir},..., \lambda^{\rm vir}_{L_{\rm vir}},2-\lambda_1^{\rm core},...,2-\lambda^{\rm core}_{L_{\rm core}}\right).
\end{align}
Rank the particle number changes (from the largest to the smallest), we obtain

\begin{equation}
\begin{aligned}
    \bm U^{\rm env (LO\mapsto MUEO)}_{L\times(L_{\rm core}+L_{\rm vir})}=\bm U^{\rm core (LO\mapsto UEO)}_{L\times L_{\rm core}}\bm U^{\rm core (UEO\mapsto MUEO)}_{L_{\rm core}\times L_{\rm core}} 
    \oplus\bm U^{\rm vir (LO\mapsto UEO)}_{L\times L_{\rm vir}}\bm U^{\rm vir (UEO\mapsto MUEO)}_{L_{\rm vir}\times L_{\rm vir}}.
\end{aligned}
\end{equation}
Finally, the coefficient matrix with particle number exchange information
\begin{align}
    \bm U^{\rm LO\mapsto MUEO}=\left(\bm U^{\rm imp(LO\mapsto MUEO)}_{L\times (L_A+L_B)}, \bm U^{\rm env (LO\mapsto MUEO)}_{L\times(L_{\rm core}+L_{\rm vir})}\right).
\end{align}

Suppose there are $N_{s}$ environment orbitals being selected, the $L\times(L_A+L_B+N_s)$ projector $P(N_s)$ can be further calculated by
\begin{align}
   P(N_s)=\bm U^{\rm LO\mapsto MUEO}_{L\times (L_A+L_B+N_s)}=\left(\bm U^{\rm imp(LO\mapsto MUEO)}_{L\times (L_A+L_B)}, \bm U^{\rm env (LO\mapsto MUEO)}_{L\times N_s}\right),
\end{align}
and the subspace is thus constructed by 
\begin{align}
    \hat{H}_{\rm sub}(N_s)=P(N_s)\hat{H}_eP^{\dagger}(N_s).
\end{align}

\section{Computing Subspace $\hat{H}_{\rm sub}(N_s)$}
\label{App-C}
The subspace $\hat{H}_{\rm sub}(N_s)$ consists of single-electron excitation term and double-electron excitation term, that is
\begin{align}
    \hat{H}_{\rm sub}(N_s)=\hat{H}_{\rm sub}^{\rm oei}(N_s)+\frac{1}{2}\hat{H}_{\rm sub}^{\rm eri}(N_s).
\end{align}
Specifically, the single excitation part
\begin{align}
\label{eq:H_sub_oei}
    \left[\hat{H}^{\rm oei}_{\rm sub}(N_s)\right]_{ij}=\sum\limits_{p,q}P_{ip}(N_s)\left(d_{pq}+(V_{\rm eff})_{pq}\right)P^{\dagger}_{qj}(N_s)
\end{align}
for index $i,j\in[L_A+L_B+N_s]$, where the exchange Coulomb interaction terms
\begin{align}
    (V_{\rm eff})_{pq}=\sum\limits_{r,s}^L D^{\rm core, LO}_{sr}\left(\langle pq|rs\rangle-\frac{1}{2}\langle pr|sq\rangle\right),
\end{align}
and 

\begin{align}
\label{eq:D_core}
    \bm{D}^{\rm core, LO}=2 \bm{U}_{L \times L_{\rm core}}^{\rm core(LO\mapsto MUEO)}\bm{U}_{L \times L_{\rm core}}^{\rm core(LO\mapsto MUEO)\dagger},
\end{align}
where the factor $2$ represents the spatial MUEO has been fully occupied by two electrons.
The double excitation part
\begin{align}
\label{eq:H_sub_eri}
    \left[\hat{H}^{\rm eri}_{\rm sub}\right]_{ijkl}=\sum\limits_{pqrs}P_{ip}(N_s)P_{jq}(N_s)\langle pq|rs\rangle P^{\dagger}_{kr}(N_s)P^{\dagger}_{ls}(N_s).
\end{align}

\section{Ground State Energy Computation}
\label{energy_compute}
The ground state energy of the full system $\hat{H}_{e}$ can be naturally computed by
\begin{align}
    E_g = E_{\rm sub}(N_s) + E_{\rm core} + E_{\rm nuc},
\end{align}
where $E_{\rm sub}(N_s)$ is the ground state energy of the subspace $H_{\rm sub}(N_s)$ provided by the quantum solver, $E_{\rm core}$ comes from the core orbitals if there are still core orbitals in MUEO and $E_{\rm nuc}$ is the nuclear repulsion energy.

The subspace energy $E_{\rm sub}(N_s)$ can be straightforwardly calculated   by the one body reduced density matrix~(1-RDM) and two body reduced density matrix~(2-RDM):
\begin{align}
    E_{\rm sub}(N_s) = \sum_{ij} {\left[\hat{H}^{\rm oei}_{\rm sub}(N_s)\right]_{ij} 
    {}^{\rm 1}D^{\rm sub}_{ij}} + 
    \frac{1}{2} \sum_{ijkl} {\left[\hat{H}^{\rm eri}_{\rm sub}(N_s)\right]_{ijkl} 
    {}^{\rm 2}D^{\rm sub}_{ijkl}},
\end{align}
where the 1-RDM
\begin{align}
   {}^{\rm 1}D^{\rm sub}_{ij}=\langle\Psi(N_s)|\hat{a}_i^{\dagger}\hat{a}_j|\Psi(N_s)\rangle,
\end{align}
and the 2-RDM
\begin{align}
   {}^{\rm 2}D^{\rm sub}_{ijkl}=\langle\Psi(N_s)|\hat{a}_i^{\dagger}\hat{a}_j^{\dagger}\hat{a}_k\hat{a}_l|\Psi(N_s)\rangle.
\end{align}
$\hat{H}^{\rm oei}_{\rm sub}(N_s)$ comes from Eq.\ref{eq:H_sub_oei} and $\hat{H}^{\rm eri}_{\rm sub}(N_s)$ comes from Eq.\ref{eq:H_sub_eri}.

The core energy $E_{\rm core}$ comes from the full occupied orbitals in MUEO basis. In detail, $E_{\rm core}$ could be calculated by using core density matrix $\bm{D}^{\rm core,LO}$~(see Eq.\ref{eq:D_core}) and the modified Fock operator $\hat{F}^{\rm LO}$~(see Eq.\ref{eq:Fock_LO}) in LO basis:
\begin{align}
    E_{\rm core} = \frac{1}{2}
    \sum_{ij}{
    (d_{ij} + \hat{F}^{\rm LO}_{ij})
    D^{\rm core,LO}_{ij}
    }.
\end{align}
Here, the introduced one body term $d_{ij}$ and a factor 1/2 can deal with the double counting problem in the exchange Coulomb interaction terms. It should be noticed that $E_{\rm core}$ is 0 if there are no core orbitals in MUEO.

\section{OE-UCCSD-VQE}
The UCCSD method is a natural quantum analogue of CCSD, which only considers correlations between occupied orbitals and unoccupied orbitals. Here, we show how to utilize OE-VQE framework to improve the performance of UCCSD method. Considering the subspace $\hat{H}_{\rm sub}(N_s)$ contains impurity and $N_s$ extension bath orbitals, and each part have been separated into occupied and unoccupied orbitals by using HF method. Naturally, the OE-UCCSD-VQE can be implemented by using single- and double- excitaion operators from occupied orbitals to unoccupied orbitals. In detail, perform the above operators to the initial state $|\Psi_{\rm init}(N_s)\rangle$, that is
\begin{align}
    |\Psi(N_s)\rangle=\prod\limits_{s\in\{pq\}}e^{\theta_{s}\hat{\tau}_s}\prod\limits_{d\in\{pqrs\}}e^{\theta_d\hat{\tau}_d}|\Psi_{\rm init}(N_s)\rangle,
\end{align}
where orbital indexes $(p,r)\in L_{\rm occ}$, $(q,s)\in L_{\rm unocc}$. The initial state $|\Psi_{\rm init}(N_s)\rangle=U_{\rm dir}(N_s-1)|\Psi_0(N_s)\rangle$, and
\begin{align}
    |\Psi_0(N_s)\rangle=\prod\limits_{m=1}^{L_{\rm occ}}\hat{a}_m^{\dagger}|0^{2L_A}\rangle\otimes\prod\limits_{l=1}^{N_s}\hat{o}_{l}|0^{N_s}\rangle
\end{align}
represents the lowest electron occupied state on the subspace $\hat{H}_{\rm sub}(N_s)$, the operator
\begin{align}
    \hat{o}_{l}=\left\{
    \begin{aligned}
    &\hat{a}_l^{\dagger}, \; {\rm if}\;  l\in{\rm core} \\
    &\hat{I}_l, \; \; {\rm if} \; l\in{\rm virt} 
    \end{aligned}
    \right.
\end{align} 
Then the convergence direction can be updated by minimizing the energy function 
\begin{align}
    \theta=\arg \min\limits_{\theta}\langle\Psi(N_s)|\hat{H}_{\rm sub}(N_s)|\Psi(N_s)\rangle.
\end{align}

\section{Proof of theorem 1}
\begin{lemma}
Consider a set of $L$-orbital Hamiltonian $T_{\bm\lambda}=\{\hat{H}|\hat{H}=UH_0U^{\dagger}\}$, where $U$ are sampled from a $2$-design unitary set, $H_0$ is a $L$-orbital Hamiltonian and $\bm\lambda=(\lambda_1,...\lambda_{2^L})$ represents the eigenvalue set of $H_0$. Then each $\hat{H}\in T_{\bm\lambda}$ shares the same eigenvalue set to $H_0$. 
\end{lemma}
\begin{proof}[Proof of lemma 2]
Suppose the spectrum decomposition of
\begin{align}
    H_0=\sum_{\lambda_i}\lambda_i|\lambda_i\rangle\langle\lambda_i|,
\end{align}
where the spectrum information $\bm\lambda=(\lambda_1,...)$. Then for any $\hat{H}\in T_{\bm\lambda}$,
\begin{align}
   \hat{H}=UH_0U^{\dagger}=\sum\limits_{\lambda_i}\lambda_iU|\lambda_i\rangle\langle\lambda_i|U^{\dagger}.
\end{align}
For any index pair $(i,j)$, the relationship $${\rm{Inner}}(U|\lambda_i\rangle,U|\lambda_j\rangle)=\langle\lambda_i|U^{\dagger}U|\lambda_j\rangle=\delta_{ij}$$
holds, then $\{U|\lambda_i\rangle\}$ formulate an orthogonal basis in the Hilbert space, and $\bm\lambda=(\lambda_1,...)$ is the spectrum information of $UH_0U^{\dagger}$.
\end{proof}
\begin{proof}[Proof of theorem 1]
In the original ADAPT-VQE framework, the gradient of energy function in terms of $\hat{\tau}_k\in P(L)$ can be expressed as
\begin{align}
    \frac{\partial E(\bm\theta)}{\partial\hat{\tau}_k}=\langle\Psi|\left[\hat{H}_e,\hat{\tau}_k\right]|\Psi\rangle,
\end{align}
where $\hat{H}_e$ represents the electron Hamiltonian, $|\Psi\rangle$ represents the reference state and $P(L)$ represents the fermionic operator pool. Then we estimate the scale of 
\begin{equation}
\begin{aligned}
    \mathbb{E}_{\hat{H}\sim S}\left[{\rm Var}_{\hat{\tau}_k\in P(L)}\left(\frac{\partial E(\bm\theta)}{\partial\hat{\tau}_k}\right)\right]&=\mathbb{E}_{\hat{H}\sim S}\left[\mathbb{E}_{\hat{\tau}_k\in P(L)}\left(\frac{\partial E(\bm\theta)}{\partial\hat{\tau}_k}\right)^2-\left(\mathbb{E}_{\hat{\tau}_k\in P(L)}\left(\frac{\partial E(\bm\theta)}{\partial\hat{\tau}_k}\right)\right)^2\right]\\
    &=\mathbb{E}_{\hat{H}\sim S}\left[\mathbb{E}_{\hat{\tau}_k\in P(L)}\left(\frac{\partial E(\bm\theta)}{\partial\hat{\tau}_k}\right)^2\right]-\mathbb{E}_{\hat{H}\sim S}\left[\left(\mathbb{E}_{\hat{\tau}_k\in P(L)}\left(\frac{\partial E(\bm\theta)}{\partial\hat{\tau}_k}\right)\right)^2\right]
\end{aligned}
\end{equation}
which characterize the fluctuations of gradient $\frac{\partial E(\bm\theta)}{\partial\hat{\tau}_k}$ around the mean-value $\mathbb{E}_{\hat{\tau}_k\in P(L)}\left(\frac{\partial E(\bm\theta)}{\partial\hat{\tau}_k}\right)$ for $\hat{H}_e\in S$. Here, $S=\{\hat{H}|{\rm Tr}(\hat{H})\leq {\rm poly}(L)\}$ is a designed electron Hamiltonian set. According to Lemma 2, we can construct $S$ based on different spectrum distribution, that is
\begin{align}
    S=\bigcup_{\bm\lambda}T_{\bm \lambda}(\hat{H}).
\end{align}
Each set $T_{\bm \lambda}(\hat{H})$ contains Hamiltonians $\hat{H}=V^{\dagger}(\hat{H})\hat{H}_0(\bm\lambda)V(\hat{H})$ with the same spectrum $\bm\lambda$, and $V$ belongs to a $2$-design unitary group $\rm U_2.$ 

We first consider the value of 
\begin{equation}
\begin{aligned}
 &\mathbb{E}_{\hat{H}\sim S}\left(\mathbb{E}_{\hat{\tau}_k\in P(L)}\left(\frac{\partial E(\bm\theta)}{\partial\hat{\tau}_k}\right)\right)^2= \mathbb{E}_{\hat{H}\sim S}\left[\frac{1}{\|P(L)\|^2}\sum\limits_{\hat{\tau}_1,\hat{\tau}_2}\langle\Psi|\left[\hat{\tau}_1,\hat{H}\right]|\Psi\rangle\langle\Psi|\left[\hat{\tau}_2,\hat{H}\right]|\Psi\rangle\right]\\
 &=\frac{1}{\|P(L)\|^2}\sum\limits_{\hat{\tau}_1,\hat{\tau}_2}\mathbb{E}_{\hat{H}\sim S}\left[\langle\Psi|\left[\hat{\tau}_1,\hat{H}\right]|\Psi\rangle\langle\Psi|\left[\hat{\tau}_2,\hat{H}\right]|\Psi\rangle\right]\\
 &=\frac{1}{\|P(L)\|^2}\sum\limits_{\hat{\tau}_1,\hat{\tau}_2}\mathbb{E}_{\hat{H}\sim S}\left[\langle\Psi|\hat{\tau}_1\hat{H}|\Psi\rangle\langle\Psi|\hat{\tau}_2\hat{H}|\Psi\rangle-\langle\Psi|\hat{\tau}_1\hat{H}|\Psi\rangle\langle\Psi|\hat{H}\hat{\tau}_2|\Psi\rangle-\langle\Psi|\hat{H}\hat{\tau}_1|\Psi\rangle\langle\Psi|\hat{\tau}_2\hat{H}|\Psi\rangle+\langle\Psi|\hat{H}\hat{\tau}_1|\Psi\rangle\langle\Psi|\hat{H}\hat{\tau}_2|\Psi\rangle\right].
\end{aligned}
\end{equation}

\begin{equation}
\begin{aligned}
&\mathbb{E}_{\hat{H}\in S}\left[\langle\Psi|\left[\hat{\tau}_1,\hat{H}\right]|\Psi\rangle\langle\Psi|\left[\hat{\tau}_2,\hat{H}\right]|\Psi\rangle\right]\\&=\bigcup\limits_{\bm\lambda}\mathbb{E}_{V}\left[{\rm Tr}\left(V(|\Psi\rangle\langle\Psi|\hat{\tau}_1)V^{\dagger}H_0\right){\rm Tr}\left(V(|\Psi\rangle\langle\Psi|\hat{\tau}_2)V^{\dagger}H_0\right)\right]
+\mathbb{E}_{V}\left[{\rm Tr}\left(V(\hat{\tau}_1|\Psi\rangle\langle\Psi|)V^{\dagger}H_0\right){\rm Tr}\left(V(|\Psi\rangle\langle\Psi|\hat{\tau}_2|)V^{\dagger}H_0\right)\right]\\
&-\mathbb{E}_{V}\left[{\rm Tr}\left(V(|\Psi\rangle\langle\Psi|\hat{\tau}_1)V^{\dagger}H_0\right){\rm Tr}\left(V(|\Psi\rangle\langle\Psi|\hat{\tau}_2)V^{\dagger}H_0\right)\right]
+\mathbb{E}_{V}\left[{\rm Tr}\left(V(|\Psi\rangle\langle\Psi|\hat{\tau}_1)V^{\dagger}H_0\right){\rm Tr}\left(V(\hat{\tau}_2|\Psi\rangle\langle\Psi|)V^{\dagger}H_0\right)\right].
\label{eq:exp1}
\end{aligned}
\end{equation}
Then the first term in~\ref{eq:exp1} can be further calculated by
\begin{equation}
\begin{aligned}
&\mathbb{E}_{V}\left[{\rm Tr}\left(V(|\Psi\rangle\langle\Psi|\hat{\tau}_1)V^{\dagger}H_0\right){\rm Tr}\left(V(|\Psi\rangle\langle\Psi|\hat{\tau}_2)V^{\dagger}H_0\right)\right]=\int_{V\in {\rm U}_2}{\rm d}\mu(V){\rm Tr}\left(V(|\Psi\rangle\langle\Psi|\hat{\tau}_1)V^{\dagger}H_0\right){\rm Tr}\left(V(|\Psi\rangle\langle\Psi|\hat{\tau}_2)V^{\dagger}H_0\right)\\
&=\frac{1}{d(d+1)}\left(\langle\Psi|\hat{\tau}_1|\Psi\rangle\langle\Psi|\hat{\tau}_2|\Psi\rangle{\rm Tr}^2(H_0)+\langle\Psi|\hat{\tau}_1|\Psi\rangle\langle\Psi|\hat{\tau}_2|\Psi\rangle+{\rm Tr}(H_0^2)\right).
\end{aligned}
\end{equation}
Similarly, the other three terms can be calculated by using properties of the group $\rm U_2$. Since $\hat{\tau}_1,\hat{\tau}_2\in\{(\hat{a}_q^{\dagger}\hat{a}_p-\hat{a}_p^{\dagger}\hat{a}_q)\}\cup\{(\hat{a}_p^{\dagger}\hat{a}_q^{\dagger}\hat{a}_r\hat{a}_s-\hat{a}_s^{\dagger}\hat{a}_r^{\dagger}\hat{a}_q\hat{a}_p)\}$ and conical fermionic anti-commutation relations $\{\hat{a}_p^{\dagger},\hat{a}_q\}=\delta_{pq}$, $\{\hat{a}_p,\hat{a}_q\}=\{\hat{a}^{\dagger}_p,\hat{a}^{\dagger}_q\}=0$,
the relationship $[\hat{\tau}_1,\hat{\tau}_2]=0$ holds, and thus
\begin{equation}
\begin{aligned}
    \mathbb{E}_{\hat{H}\sim S}\left[\mathbb{E}_{\hat{\tau}_k\in P(L)}\left(\frac{\partial E(\bm\theta)}{\partial\hat{\tau}_k}\right)\right]^2=0.
\end{aligned}
\end{equation}
Then we consider another term
\begin{equation}
\begin{aligned}
    \mathbb{E}_{\hat{H}\in S}\left[\mathbb{E}_{\hat{\tau}\in P(L)}\left(\frac{\partial E(\bm\theta)}{\partial\hat{\tau}_k}\right)^2\right]&=\frac{1}{\|P(L)\|^2}\sum\limits_{\hat{\tau}_k}\mathbb{E}_{\hat{H}\in S}\left[\langle\Psi|\left[\hat{H}_e,\hat{\tau}_k\right]|\Psi\rangle^2\right].
\end{aligned}
\end{equation}
Using the relationship $\hat{H}=V^{\dagger}(\hat{H})\hat{H}_0V(\hat{H})$,
\begin{equation}
\begin{aligned}
\mathbb{E}_{\hat{H}\in S}\left[\langle\Psi|\left[\hat{H}_e,\hat{\tau}_k\right]|\Psi\rangle^2\right]&=\mathbb{E}_{V}\left[{\rm Tr}^2\left(V(\hat{\tau}_k|\Psi\rangle\langle\Psi|)V^{\dagger}H_0\right)\right]+\mathbb{E}_{V}\left[{\rm Tr}^2\left(V(|\Psi\rangle\langle\Psi|\hat{\tau}_k)V^{\dagger}H_0\right)\right]\\
&-2\mathbb{E}_V\left[{\rm Tr}\left(V(|\Psi\rangle\langle\Psi|\hat{\tau}_k)V^{\dagger}H_0\right){\rm Tr}\left(V(\hat{\tau}_k|\Psi\rangle\langle\Psi|)V^{\dagger}H_0\right)\right].
\end{aligned}
\end{equation}
Since the unitary transformation does not change the spectral distribution for all $\hat{H}\in S_{\bm\lambda}(\hat{H})$, then
\begin{equation}
\begin{aligned}
&\mathbb{E}_{V}\left[{\rm Tr}^2\left(V(\hat{\tau}_k|\Psi\rangle\langle\Psi|)V^{\dagger}H_0\right)\right]+\mathbb{E}_{V}\left[{\rm Tr}^2\left(V(|\Psi\rangle\langle\Psi|\hat{\tau}_k)V^{\dagger}H_0\right)\right]\\
&=\int_{V\in {\rm U}_2}{\rm d}\mu(V){\rm Tr}^2\left(V(\hat{\tau}_k|\Psi\rangle\langle\Psi|)V^{\dagger}H_0\right)+\int_{V\in {\rm U}_2}{\rm d}\mu(V){\rm Tr}^2\left(V(|\Psi\rangle\langle\Psi|\hat{\tau}_k)V^{\dagger}H_0\right)\\
&=\frac{2}{d^2-1}\left(\langle\Psi|\hat{\tau}_k|\Psi\rangle^2{\rm Tr}^2(H_0)+\langle\Psi|\hat{\tau}_k|\Psi\rangle^2+{\rm Tr}(H_0^2)\right)
-\frac{2}{d(d^2-1)}\left(\langle\Psi|\hat{\tau}_k|\Psi\rangle^2{\rm Tr}^2(H_0)+\langle\Psi|\hat{\tau}_k|\Psi\rangle^2{\rm Tr}(H_0^2)\right),
\end{aligned}
\end{equation}
and 
\begin{equation}
\begin{aligned}
&2\mathbb{E}_V\left[{\rm Tr}\left(V(|\Psi\rangle\langle\Psi|\hat{\tau}_k)V^{\dagger}H_0\right){\rm Tr}\left(V(\hat{\tau}_k|\Psi\rangle\langle\Psi|)V^{\dagger}H_0\right)\right]\\
&=2\int_{V\in {\rm U}_2}{\rm d}\mu(V){\rm Tr}\left(V(|\Psi\rangle\langle\Psi|\hat{\tau}_k)V^{\dagger}H_0\right){\rm Tr}\left(V(\hat{\tau}_k|\Psi\rangle\langle\Psi|)V^{\dagger}H_0\right)\\
&=\frac{2}{d^2-1}\left(\langle\Psi|\hat{\tau}_k|\Psi\rangle^{2}{\rm Tr}^2(H_0)+\langle\Psi|\hat{\tau}^2_k|\Psi\rangle+{\rm Tr}(H_{0}^{2})\right)-\frac{2}{d(d^2-1)}\left(\langle\Psi|\hat{\tau}^2_k|\Psi\rangle{\rm Tr}^2(H_0)+\langle\Psi|\hat{\tau}_k|\Psi\rangle^2{\rm Tr}(H_{0}^{2})\right),
\end{aligned}
\end{equation}
where $\mu(V)$ represents the Haar-measure in the $2$-design unitary group $\rm U_2$. Combine the above two equations, we can upper bound
\begin{equation}
\begin{aligned}
\mathbb{E}_{\hat{H}\in S}\left[\langle\Psi|\left[\hat{H}_e,\hat{\tau}_k\right]|\Psi\rangle^2\right]=\frac{\Delta(\hat{\tau}_k,|\Psi\rangle)}{4^n-1}\left(1-\frac{{\rm Tr}^2(H_0)}{2^n}\right)\leq\frac{\Delta(\hat{\tau}_k,|\Psi\rangle)}{4^n-1}.
\end{aligned}
\end{equation}
Here, $\Delta(\hat{\tau}_k,|\Psi\rangle)=\langle\Psi|\hat{\tau}_k|\Psi\rangle^2-\langle\Psi|\hat{\tau}_k^2|\Psi\rangle$ represents the fluctuation on using $\hat{\tau}_k$ to measure $|\Psi\rangle$. 

Since the anti-hermitian operator $\hat{\tau}_k$ acts on different orbital domains to characterize annihilation and creation behaviours, and it can be encoded into linear combinations of $n$-qubit Pauli operators, that is $\hat{\tau}_k\mapsto i\sum_{j=1}^mP_j^{k}$, where $m\in\{4, 8\}$. Therefore
\begin{equation}
\begin{aligned}
\Delta(\hat{\tau}_k,|\Psi\rangle)&=\sum\limits_{j_1,j_2=1}^m\left(\langle\Psi|P_{j_1}^kP_{j_2}^k|\Psi\rangle-\langle\Psi|P_{j_1}^k|\Psi\rangle\langle\Psi|P_{j_2}^k|\Psi\rangle\right)< 2m^2.
\end{aligned}
\end{equation}
The last inequality comes from the fact that $P_{j_1}^kP_{j_2}^k, P_{j_1}^k$ and $P_{j_2}^k$ are Pauli operators, and their eigenvalues thus belongs to $\{-1, +1\}$. Therefore, $\langle\Psi|P_{j_1}^kP_{j_2}^k|\Psi\rangle-\langle\Psi|P_{j_1}^k|\Psi\rangle\langle\Psi|P_{j_2}^k|\Psi\rangle<2$ holds.
\end{proof}

\section{Proof of Corollary 1.1}
\begin{proof}
Given an electron Hamiltonian $\hat{H}_e$ with $L$ orbitals, and $M_{\rm OE-Adapt}$ ($M_{\rm Adapt}$) represent the quantum measurement complexity of OE-ADAPT-VQE (ADAPT-VQE). Suppose there are $K$ fermionic operators are required to construct the ground state, then according to Theorem~\ref{theorem1}, one has $M_{\rm adapt}= \mathcal{O}\left(KL^44^L\right)$. On the other hand,
\begin{equation}
\begin{aligned}
    M_{\rm OE-adapt}&=\sum\limits_{i=L_A+L_B}^L\left(\frac{K}{L-L_A-L_B}\right)i^44^{i}\\
    &\leq\frac{K}{L-L_A-L_B}\int_{L_A+L_B}^Lx^4e^{\ln 4 x}dx\\
    &<\frac{KL^4}{(L-L_A-L_B)\ln 4}4^L=\frac{M_{\rm adapt}}{(L-L_A-L_B)\ln 4}.
\end{aligned}
\end{equation}
\end{proof}

\end{document}